\documentclass[letterpaper,11pt]{article}

\usepackage{graphicx,epsf,psfrag,amsmath,amssymb}

\newtheorem{theorem}{Theorem}[section]

\newtheorem{corollary}{Corollary}[section]
\newtheorem{definition}{Definition}[section]

\newcommand{\qed}{\hfill\hbox{\rlap{$\sqcap$}$\sqcup$}}
\newenvironment{proof}{\noindent \emph{Proof.\,}}{\qed}

\def\etal{{\it et~al.}\,}
\newcommand{\frechet}{Fr\'echet}
\newcommand{\dfre}{d_{F}}

\usepackage{fancyhdr}
\newcommand{\myfootnote}[2]{%
	\thispagestyle{fancy}
	\fancyhf{}
	\renewcommand{\headrulewidth}{0pt}
	\lfoot{\footnotesize\emph{#1}}
	\rfoot{\footnotesize\emph{#2}}
}
\begin{document}

\begin{titlepage}  

\title{On the Complexity of Protein Local Structure Alignment Under the Discrete Fr\'{e}chet Distance}
\author{
Binhai Zhu
\thanks{Department of Computer Science, Montana State University, Bozeman, MT 59717-3880, USA. Email: {\tt bhz@cs.montana.edu}.}
}

\date{}
\maketitle
\myfootnote{Manuscript}{September 5, 2007}

\begin{abstract}

Protein structure alignment is a fundamental problem in computational
and structural biology. While there has been lots of experimental/heuristic
methods and empirical results, very little is known regarding the
algorithmic/complexity aspects of the problem, especially on protein local
structure alignment. A well-known measure to characterize the similarity of
two polygonal chains is the famous Fr\'{e}chet distance and with the
application of protein-related research, a related discrete Fr\'{e}chet
distance has been used recently. In this paper,
following the recent work of Jiang, {\em et al.} we investigate the protein
local structural alignment problem using bounded discrete Fr\'{e}chet distance.
Given $m$ proteins (or protein backbones, which are 3D polygonal chains),
each of length $O(n)$, our main results are summarized as follows.
\begin{itemize}
\item If the number of proteins, $m$, is not part of the input, then the
problem is NP-complete; moreover, under bounded discrete Fr\'{e}chet distance
it is NP-hard to approximate the maximum size common local structure
within a factor of $n^{1-\epsilon}$. These results hold both when all
the proteins are static or when translation/rotation are allowed.
\item If the number of proteins, $m$, is a constant, then there is a polynomial
time solution for the problem.
\end{itemize}

\end{abstract}

\noindent
{\bf Keywords}: Protein structure alignment, Fr\'{e}chet distance, Discrete Fr\'{e}chet
distance, Approximation, NP-hardness
\end{titlepage}
\newpage

\section{Introduction}

As a famous distance measure in the field of abstract spaces, Fr\'{e}chet
distance was first defined by Maurice Fr\'{e}chet a century ago \cite{Fr06}.
Alt and Godau first used it in measuring the similarity of polygonal chains
in 1992 \cite{AG92}. It is well known that the Fr\'{e}chet distance between
two two-dimensional (2D) polygonal chains (polylines) can be computed in
polynomial time \cite{AG92,AG95}, and even under translation or rotation
(though the running time is much higher) \cite{AKW01}.
In three-dimensional space (3D), Wenk should that given two chains with
sum of length $N$, the minimum Fr\'{e}chet distance between them can be computed
in $O(N^{3f+2}\log N)$ time, where $f$ is the degree of freedom for moving
the chains \cite{We02}.
So with translation alone this minimum Fr\'{e}chet distance can be computed in
in $O(N^{11}\log N)$ time, and when both translation and rotation are
allowed the corresponding minimum Fr\'{e}chet distance can be computed in
$O(N^{20}\log N)$ time. These results can be generalized to any fixed
dimensions \cite{We02}. While computing (approximating) Fr\'{e}chet distance
for surfaces is in general NP-hard \cite{Go98,HS04}, it is polynomially solvable
for restricted surfaces \cite{BBW06}.

In 1994, Eiter and Mannila defined the {\em discrete Fr\'{e}chet distance}
between two polygonal chains $A$ and $B$ (in any fixed dimensions) and it
turns out that this simplified distance is always realized by two vertices
in $A$ and $B$ \cite{EM94}. They also showed that with dynamic programming
the discrete Fr\'{e}chet distance between them can be computed in $O(|A||B|)$
time.

Recently, Jiang, Xu and Zhu applied the discrete Fr\'{e}chet distance in
(globally) aligning the backbones of proteins (which is called the {\em protein
structure-structure alignment} or more generally, the {\em protein global
alignment} problem) \cite{JXZ07}. In fact, in this
application the discrete Fr\'{e}chet distance makes more sense as the backbone
of a protein is simply a polygonal chain in 3D, with each vertex being the
alpha-carbon atom of a residue. So if the (continuous) Fr\'{e}chet distance
is realized by an alpha-carbon atom and some other point which does not
represent an atom, it is not meaningful biologically. Jiang, {\em et al.}
showed that given two 2D (or 3D) polygonal chains the minimum discrete
Fr\'{e}chet distance between them, under both translation and rotation, can be
computed in polynomial time. They also applied some ideas therein to design an
efficient heuristic for the original protein structure-structure alignment
problem in 3D and the empirical results showed that their alignment is more
accurate compared with previously known solutions.

In essence, the result of Jiang, Xu and Zhu \cite{JXZ07} implies that the
protein global alignment problem, which is to find all proteins in a given
set ${\cal P}$ similar to a query protein or some protein in ${\cal P}$
(under translation and rotation), is
polynomially solvable. However, very little algorithmic/complexity results
is known regarding the protein local structure alignment problem. The only 
such recent result was due to Qian, {\em et al.} who showed that under
the RMSD distance the problem is NP-complete but admits a PTAS \cite{QL+07}.
On the other hand, there have been lots of experimental/heuristic methods with
practical systems since 1989, e.g., SSAP \cite{TO89}, DALI \cite{HS93,HP00}, CATH \cite{OM+97},
CE \cite{SB98}, SCOP \cite{CA+00}, MAMMOTH \cite{OS+02} and TALI \cite{MWV08}.
In this paper, we show that if many proteins are given then the local structure
alignment problem, under the discrete Fr\'{e}chet distance, is very hard; on
the other hand, if only a small number of proteins are given then there is a
polynomial time solution for the problem.

The paper is organized as follows. In Section 2, we introduce some basic
definitions regarding \frechet\ distance and review some known results.
In Section 3, we show the hardness result for the protein local structure
alignment problem.
In Section 4, we show how to solve the problem when $m$
is a constant. In Section 5, we conclude the paper with
several open problems.

\section{Preliminaries}

Given two 3D polygonal chains $A,B$ with $|A|=k$ and $|B|=l$ vertices respectively,
we aim at measuring the similarity of $A$ and $B$ (possibly under translation
and rotation) such that their distance is minimized under certain measure.
Among the various distance measures, the Hausdorff distance is known to be
better suited for matching two point sets than for matching two polygonal
chains; the (continuous) \frechet\ distance is a superior measure for matching
two polygonal chains, but it is not quite easy to compute \cite{AG92}.

Let $X$ be the Euclidean space $\mathbb{R}^3$; let $d(a,b)$ denote the
Euclidean distance between two points $a,b\in X$.
The (continuous) \frechet\ distance between two parametric curves $f:[0,1]\to X$
and $g:[0,1]\to X$ is 
$$
\delta_\mathcal{F}(f,g) = \inf_{\alpha,\beta} \max_{s\in[0,1]}
d(f(\alpha(s)),g(\beta(s))),
$$
where $\alpha$ and $\beta$ range over all continuous non-decreasing real
functions with $\alpha(0) = \beta(0) = 0$ and $\alpha(1) = \beta(1) = 1$
\footnote{This definition holds in any fixed dimensions.}.

Imagine that a person and a dog walk along two different paths while connected
by a leash; moreover, they always move forward, possibly at different paces.
Intuitively, the minimum possible length of the leash is the \frechet\ distance between the two paths.
To compute the \frechet\ distance between two polygonal curves $A$ and $B$
(in the Euclidean plane) of $|A|$ and $|B|$ vertices, respectively,
Alt and Godau \cite{AG92} presented an $O(|A||B|\log^2(|A||B|))$ time algorithm.
Later this bound was reduced to $O(|A||B|\log(|A||B|))$ time \cite{AG95}.

We now define the discrete \frechet\ distance following \cite{EM94}.

\begin{definition}
Given a polygonal chain (polyline) in 3D, $P=$ $\langle p_1,\dots,p_k\rangle$ of $k$
vertices, a \textbf{$m$-walk} along
$P$ partitions the path into $m$ disjoint non-empty subchains
$\{{\cal P}_i\}_{i=1..m}$ such that ${\cal P}_i=$ $\langle p_{k_{i-1}+1},\dots,p_{k_i}\rangle$
and $0 = k_0 < k_1 < \dots < k_{m} = k$.

Given two 3D polylines $A=$ $\langle a_1,\dots,a_k\rangle$ and $B=$ $\langle b_1,\dots,b_l\rangle$,
a \textbf{paired walk} along $A$ and $B$ is
a $m$-walk $\{{\cal A}_i\}_{i=1..m}$ along $A$ and
a $m$-walk $\{{\cal B}_i\}_{i=1..m}$ along $B$ for some $m$, such that,
for $1 \le i \le m$, either $|{\cal A}_i| = 1$ or $|{\cal B}_i| = 1$
(that is, ${\cal A}_i$ or ${\cal B}_i$ contains exactly one vertex).
The \textbf{cost} of a paired walk
$W = \{({\cal A}_i,{\cal B}_i)\}$ along two paths $A$ and $B$ is
$$
\dfre^W(A,B) = \max_i \max_{(a,b) \in {\cal A}_i \times {\cal B}_i} d(a,b).
$$

The \textbf{discrete \frechet\ distance} between two polylines $A$ and $B$ is
$$
\dfre(A,B) = \min_W \dfre^W(A,B).
$$
The paired walk that achieves the discrete \frechet\ distance between two paths
$A$ and $B$ is also called the \textbf{\frechet\ alignment} of $A$ and $B$.
\end{definition}

Consider the scenario in which the person walks (jumps) along $A$ and the dog along $B$.
Intuitively, the definition of the paired walk is based on three cases:
\begin{enumerate}
\item $|{\cal B}_i| > |{\cal A}_i| = 1$:
the person stays and the dog moves (jumps) forward;
\item $|{\cal A}_i| > |{\cal B}_i| = 1$:
the person moves (jumps) forward and the dog stays;
\item $|{\cal A}_i| = |{\cal B}_i| = 1$:
both the person and the dog move (jump) forward.
\end{enumerate}

\begin{figure}[hbt]
\centerline{\epsffile{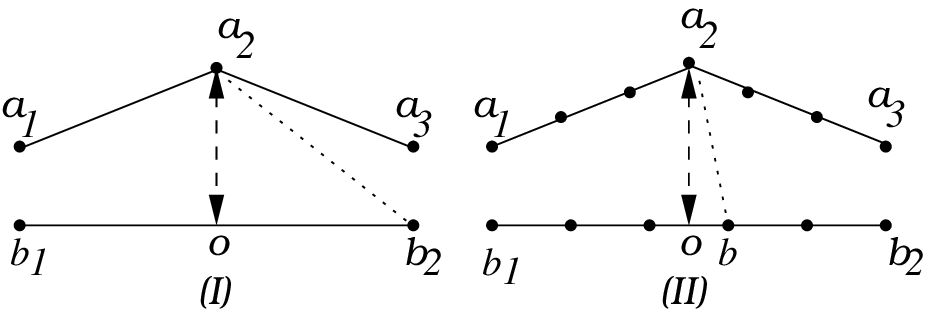}}
\begin{center}{\small {\bf Fig. 1}. The relationship between
discrete and continuous \frechet\ distances.}
\end{center}
\end{figure}

Eiter and Mannila presented a simple dynamic programming algorithm
to compute $\dfre(A,B)$ in $O(|A||B|)=O(kl)$ time \cite{EM94}. Recently,
Jiang, \etal\ showed that the minimum discrete \frechet\ distance between
two chains in 2D, $A$ and $B$, under translation can be computed in
$O(k^3l^3\log(k+l))$ time, and under both translation and rotation it can be
computed in $O(k^4l^4\log(k+l))$ time \cite{JXZ07}. For 3D chains
these bounds are $O(k^4l^4\log(k+l))$ and $O(k^7l^7\log(k+l))$ respectively
\cite{JXZ07}. They are significantly faster than the corresponding bounds
for the continuous \frechet\ distance (certainly due to a simpler
distance structure), which are $O((k+l)^{11}\log(k+l))$ and 
$O((k+l)^{20}\log(k+l))$ respectively for 3D chains \cite{We02}.

We comment that while the discrete \frechet\ distance could be arbitrarily
larger than the corresponding continuous \frechet\ distance (e.g., in Fig.~1 (I),
they are $d(a_2,b_2)$ and $d(a_2,o)$ respectively), by adding sample points on the
polylines, one can easily obtain a close approximation of the continuous
\frechet\ distance using the discrete \frechet\ distance (e.g., one can
use $d(a_2,b)$ in Fig.~1 (II) to approximate $d(a_2,o)$). This fact was
pointed before in \cite{EM94,In02} and is supported by the fact
that the segments in protein backbones are mostly of similar lengths.
Moreover, the discrete \frechet\ distance
is a more natural measure for matching the geometric shapes of biological
sequences such as proteins. As we mentioned in the introduction, in such an
application, continuous \frechet\ does not make much sense to biologists.

In the remaining part of this paper, for the first time, we investigate the
locally aligning a set of polygonal chains (proteins or protein backbones)
in 3D, under the discrete Fr\'{e}chet distance. 

\section{Protein Local Structure Alignment is Hard}

Given a set of proteins modeled as simple 3D polygonal chains,
the Protein Local Structure Alignment (PLSA) problem is defined as follows.

\noindent
{\bf Instance:} Given a set $m$ of proteins $P_1$, $P_2$, ..., $P_m$
in 3D, each with length $O(n)$, and a real number $D$.\\
{\bf Problem:} Does there exist a chain $C$ of $k$ vertices such that
the vertices of $C$ are from $P_i$'s, and $C$ and a subsequence of $P_i$
($1\leq i\leq m$) has discrete Fr\'{e}chet distance at most $D$
(under translation and rotation)? \\

If no translation and rotation is allowed, we call the corresponding
problem {\em static} PLSA. For the optimal version of the problem, we
wish to maximize $k$ when $D$ is given. The (polynomial-time) approximation
solution will also be referred to as approximating the optimal solution
value $k^*$ when it is hard to compute exactly. We will see that it is also
hard to approximate $k^*$ even for static PLSA. We first prove the following
theorem.

\begin{theorem}
Given $D=\delta$, the static PLSA problem does not admit any approximation
of factor $n^{1-\epsilon}$ unless P=NP.
\end{theorem}

\begin{proof}
It is easy to see that PLSA belongs to NP.
We use a reduction from Independent Set to the Protein Local Structure
Alignment Problem. Independent Set is a well known NP-complete
problem which cannot be approximated within a factor of $n^{1-\epsilon}$
\cite{Ha99}. The general idea is similar to that of the longest common
subsequence problem for multiple sequences \cite{JL95}, but our details are
much more involved due to the geometric properties of the problem.

Given a graph $G=(V,E), V=\{v_1,v_2,\cdots,v_N\}, E=\{e_1,e_2,\cdots,e_M\}$,
we construct $M+1$ 3D chains $P_0,P_1,P_2,...,P_M$ as follows. (We assume
that the vertices and edges in $G$ are sorted by their corresponding indices.)

The overall reduction is as follows: ${\cal P}=\{P_0,P_1,P_2,...,P_M\}$, and
$$P_0=\langle v'_1,v'_2,\cdots, v'_n\rangle,$$ where
$v'_i=(i,i^2,0)$ is a 3D point for $i=1,...,n$.

\noindent
For each $e_r=(v_i,v_j)$ in $G$, we have a corresponding sequence (3D chain)
$$P_r=\langle v'_1,v'_2,\cdots, v'_{i-1},v'_{i+1},\cdots, v'_n,v"_1,v"_2,\cdots, v"_{j-1},v"_{j+1},\cdots, v"_n\rangle,$$ where
$v'_i=(i,i^2,0)$ and $v"_i=(i,i^2,\delta)$ are 3D points for $i=1,...,n$
and $\delta$ is an arbitrarily small positive real number less than 0.1.
 
\noindent
We claim that $G$ has an independent set of size $k$ if and only if there is
a chain $C$ of $k$ vertices such that the discrete Fr\'{e}chet distance
between $C$ and a subsequence of $P_r$, $S_r$, is at most $\delta$ (i.e.,
$d_F(C,S_r)\leq \delta$).
The following claims are made with the detailed proofs left out. 
 
{\bf Claim A.} $P_r$ is a simple polygonal chain in 3D.

{\bf Claim B.} $S_r$ is a simple polygonal chain in 3D with $|S_r|=k$.

If $G$ has an independent set of size $k$, then the chain $C$ can be
constructed as follows. Let the independent set of $G$ be
ordered as $I=\langle v_{i_1},v_{i_2},...,v_{i_k}\rangle$ with $i_1<i_2<...<i_k$.
For $r=0,1,...,M$, we scan $P_r$ in a greedy fashion to obtain the first
$v'_j$ or $v"_j$ such that the first component of its coordinate is $i_1$.
Repeat this process to obtain $S_r$. Then let any $S_r$ be $C$. Obviously,
$C$ has $k$ vertices and $|S_r|=k$ for $r=0,1...,M$.

If there is a chain $C$ of $k$ vertices such that the discrete Fr\'{e}chet
distance between $C$ and a subsequence of $P_r$, $S_r$, is at most $\delta$
(i.e., $d_F(C,S_r)\leq \delta)$, then we can see the following.

\noindent
{\bf Property (a)} Let $P_r=\langle v'_1,v'_2,\cdots, v'_{i-1},v'_{i+1},\cdots, v'_n,v"_1,v"_2,\cdots, v"_{j-1},v"_{j+1},\cdots, v"_n\rangle$, then $d(v'_{p},v"_q)>3$ for all $p\neq q$.

\noindent
{\bf Property (b)} Let $P_r=\langle v'_1,v'_2,\cdots, v'_{i-1},v'_{i+1},\cdots, v'_n,v"_1,v"_2,\cdots, v"_{j-1},v"_{j+1},\cdots, v"_n\rangle$, then $d(v'_{p},v"_p)\leq \delta$ for all $p\neq i,p\neq j$.

\noindent
{\bf Property (c)} Let $P_r=\langle u_1,u_2,\cdots, u_{O(n)}\rangle$, then
$|d(u_{p},u_{q})-d(u_{p'},u_{q'})|>>\delta$ as long as the first
components of the 4 coordinates of $u_{p},u_{q},u_{p'},u_{q'}$ are
all different.

As $\delta$ is very small, when $d_F(C,S_r)\leq \delta$, the vertices
of $C$ and $S_r$ must be matched orderly in a one-to-one fashion.
(In other words, the man walking on $C$ and the dog walking on $S_r$ must
move/jump together at each vertex. Otherwise, $d_F(C,S_r)>3>>\delta$.) We
now claim that the (ordered) vertices of $C$ correspond to an independent
set $I$ of $G$; moreover, if $C=\langle C_1,C_2,\cdots,C_k\rangle$ and $C_p=(x_p,y_p,z_p)$,
then $v_{x_p}\in I$. Suppose that $C_p=(x_p,y_p,z_p)$, $C_q=(x_q,y_q,z_q)$
and $v_{x_p},v_{x_q}\in I$ but there is an edge $e_t=(v_{x_p},v_{x_q})\in E$.
By our construction of $P_t$ (from $e_t$), $v'_{x_p}$ and $v"_{x_q}$ are not
included in $P_t$ and $v'_{x_q}$ precedes $v"_{x_p}$ in $P_t$. This is a
contradiction.

To conclude the proof of this theorem, notice that the reduction
take $O(MN)$ time.
\end{proof}

\begin{figure}[hbt]
\centerline{\epsffile{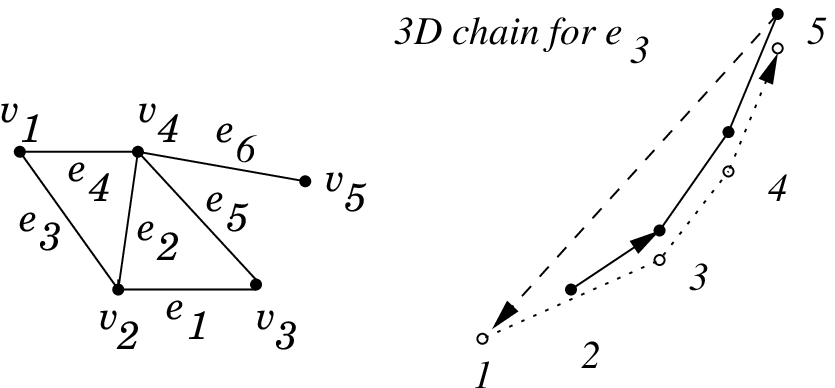}}
\label{}
\begin{center}{\bf Figure 1. Illustration of a simple graph for the reduction.}
\end{center}
\end{figure}

In the example shown in Figure 1, we have

$P_1=\langle v'_1,v'_3,v'_4,v'_5,v"_1,v"_2,v"_4,v"_5\rangle$,

$P_2=\langle v'_1,v'_3,v'_4,v'_5,v"_1,v"_2,v"_3,v"_5\rangle$,

$P_3=\langle v'_2,v'_3,v'_4,v'_5,v"_1,v"_3,v"_4,v"_5\rangle$,

$P_4=\langle v'_2,v'_3,v'_4,v'_5,v"_1,v"_2,v"_3,v"_5\rangle$,

$P_5=\langle v'_1,v'_2,v'_4,v'_5,v"_1,v"_2,v"_3,v"_5\rangle$, and

$P_6=\langle v'_1,v'_2,v'_3,v'_5,v"_1,v"_2,v"_3,v"_4\rangle$.

\noindent
An example of $P_3$ is shown in Figure 1 as well, in which case
black nodes are on the $Z=0$ plane and white nodes are on the
$Z=\delta$ plane (apparently for the visualization reason, the XY-plane
is slanted). The solid segments are on the
$Z=0$ plane, the dotted segments are on the $Z=\delta$ plane and
the only dashed segment connects two points on different planes.
Corresponding to the optimal independent set $\{v_1,v_3,v_5\}$ in $G$,
the optimal local alignment $C=\langle v'_1,v'_3,v'_5\rangle$ matches $P_3$ at its
subsequence $S_3=\langle v"_1,v"_3,v"_5\rangle$. 

\begin{corollary}
Given $D=\delta$ and when both translation and rotation are allowed, the
(maximization version of) PLSA problem does not admit any approximation
of factor $n^{1-\epsilon}$ unless P=NP.
\end{corollary}

\begin{proof}
Due to Property (a), (b) and (c), translation/rotation will not be able to
generate another $C'$ which is topologically different from $C$.
\end{proof}

Notice that in our proof all the adjacent vertices in $C$ could be non-adjacent
in $P_i$, for $i=0,1,...,m$. Biologically, this might be a problem as one
residue alone sometimes cannot carry out any biological function.
Define a $c$-{\em substring} or a $c$-{\em subchain} of $P_i$ as a
continuous subchain of $P_i$ with at least $c$ vertices. Unfortunately, even
if we introduce this condition by forcing that $C$ is composed of $k$ ordered
$c$-substrings of each $P_i$, for some constant $c$, the above proof can be
modified to maintain a valid reduction from Independent Set. Call this
corresponding problem Protein $c$-Local Structure Alignment (PcLSA), in
which $C$ must be composed of $k$ ordered $c$-subchains of each $P_i$.
We have the following corollary.

\begin{corollary}
The maximization version of PcLSA does not admit any approximation
of factor $n^{1-\epsilon}$ unless P=NP.
\end{corollary}

\section{Polynomial Time Solutions for PLSA When $m$ is Small}
In this section, we present a polynomial time solution for the 
PLSA problem when $m$ is a constant. We first show a dynamic programming
solution for the static PLSA and then we show how to use that as a subroutine
for the general PLSA problem, when $m$ is small.
\subsection{A Dynamic Programming Solution for the Static PLSA When $m$ is Small}

In this subsection, we present a dynamic programming solution for the static
PLSA problem when $m$ is small. Such a solution can be used as a subroutine
for the general PLSA problem. We first consider the case when $m=2$.
Besides $C$, we try to maximize the length of the aligned subsequences
in $P_1=A$ and $P_2=B$ with $|A|=n_1,|B|=n_2$. For the ease of description,
we only show how to obtain these lengths which are stored in $D[-,-,-,-]$
and $M[-,-,-,-]$ respectively. It is easy to reconstruct $C$ from these arrays.

Let $A[i_1,i_2]$ be a subchain of $A$ starting from the index $i_1$ and ending
at the index $i_2$.
Let $B[j_1,j_2]$ be a subchain of $B$ starting from the index $j_1$ and ending
at the index $j_2$.
$D[i_1,i_2,j_1,j_2]$ stores the length of the aligned subsequences of
$A[i_1,i_2]$ as a consequence of the alignment of $C$ and $A[i_1,i_2]$,
and $C$ and $B[j_1,j_2]$. $M[i_1,i_2,j_1,j_2]$ is defined symmetrically.

Intuitively $D[-,-,-,-]$ stores the length of aligned
subsequences from chain $A$ (dog's route) and $M[-,-,-,-]$ stores the length
of aligned subsequences from chain $B$ (man's route). Define
$T_{F}(i_1,i_2,j_1,j_2)$ as the sum of aligned subsequences in
both $A[i_1,i_2]$ and $B[j_1,j_2]$. Writing $A[i]$ as $a_i$ and $B[j]$ as $b_j$,
we have the dynamic programming solution as follows.
$$
T_{F}(i_1,i_2,j_1,j_2) = D(i_1,i_2,j_1,j_2)+M(i_1,i_2,j_1,j_2),
$$

where 
\begin{equation} 
D(i_1,i_2,j_1,j_2)=\max\begin{cases}
\max_{i_1\leq k_1 < i_2} \{D(i_1,k_1,j_1,j_2) + 1\}& \text{ if $d(a_{i_2},b_{j_2})\leq\delta, \backslash\backslash$ dog moves}\\
\max_{i_1\leq k_1 < i_2,j_1\leq k_2 < j_2} \{D(i_1,k_1,j_1,k_2)+1\}& \text{ if $d(a_{i_2},b_{j_2})\leq\delta,\backslash\backslash$ both move}\\
\max_{j_1\leq k_2 < j_2} \{D(i_1,i_2,j_1,k_2)\}& \text{ if $d(a_{i_2},b_{j_2})\leq\delta, \backslash\backslash$ dog stays}
\end{cases}
\end{equation}
and
\begin{equation} 
M(i_1,i_2,j_1,j_2)=\max\begin{cases}
\max_{i_1\leq k_1 < i_2} \{M(i_1,k_1,j_1,j_2)\}& \text{ if $d(a_{i_2},b_{j_2})\leq\delta, \backslash\backslash$ man stays}\\
\max_{i_1\leq k_1 < i_2,j_1\leq k_2 < j_2} \{M(i_1,k_1,j_1,k_2)+1\}& \text{ if $d(a_{i_2},b_{j_2})\leq\delta, \backslash\backslash$ both move}\\
\max_{j_1\leq k_2 < j_2} \{M(i_1,i_2,j_1,k_2)+1\}& \text{ if $d(a_{i_2},b_{j_2})\leq\delta, \backslash\backslash$ man moves}
\end{cases}
\end{equation}

The boundary cases are handled as follows.

\begin{equation}
D(i_1,i_1,j_1,j_1)=M(i_1,i_1,j_1,j_1)=\begin{cases}
1 & \text{ if $d(a_{i_1},b_{j_1})\leq\delta$,}\\
0 & \text{ if $d(a_{i_1},b_{j_1})>\delta$.}
\end{cases}
\end{equation}

The final solution value is stored in $T_F[1,n_1,1,n_2]$. We have the following theorem.
\begin{theorem}
When $m=2$, the static PLSA problem can be solved in $O(n^4)$ time and space.
\end{theorem}

It is easy to generalize this algorithm to the more general case when $m$ is
some constant. We thus have the following corollary.
 
\begin{corollary}
When $m$ is a constant, the static PLSA problem can be solved in $O(m^3n^{2m})$
time and $O(mn^{2m})$ space.
\end{corollary}

\subsection{A Polynomial Time Solution for PLSA When $m$ is Small}

Apparently, for any solution for PLSA we should allow translation and rotation.
When $m=2$ and when both translation and rotation are allowed, we can use
a method similar to that in \cite{JXZ07} to compute the optimal local alignment
with fixed $\delta$. The idea is as follows. Without loss of generality,
we assume that $A$ is static and we translate/rotate $B$ and let $\tau(B)$ be
the copy of $B$ after some translation/rotation. Let $|A|=n_1,|B|=n_2$ and let
$f$ be the degree of freedom for moving $B$. As we are in 3D and both
translation and rotation are allowed, we have $f=6$. We can enumerate all
possible configurations for $A$ and $\tau(B)$ to realize a discrete Fr\'{e}chet
distance of $\delta$. There are $O((n_1n_2)^f)=O(n^{12})$ number of such
configurations, following an argument similar to \cite{We02,JXZ07}. Then for
each configuration, we can use the above Theorem 4.1 to obtain the optimal
local alignment for each configuration and finally we simply return the
overall optimal solution.

\begin{corollary}
When $m=2$ and when both translation and rotation are allowed,
the PLSA problem can be solved in $O(n^{16})$ time and $O(n^{4})$ space.
\end{corollary}

We comment that when $m$ is larger, but still a constant, the above idea
can be carried over so that we will still be able to solve PLSA in polynomial
time. It follows from \cite{We02,JXZ07} that we have $O(n^{mf})=O(n^{6m})$
number of configurations between the $m$ chains. Then we can again use
Corollary 4.1 to obtain the optimal local alignment for each configuration.
The overall complexity would be $O(n^{6m}\times m^3n^{2m})=O(m^3n^{8m})$ time
and $O(mn^{2m})$ space. Certainly, such an algorithm is only meaningful in
theory.

\begin{corollary}
When $m$ is a constant and when both translation and rotation are allowed,
the PLSA problem can be solved in $O(m^3n^{8m})$ time and $O(mn^{2m})$ space.
\end{corollary}

\section{Concluding Remarks}

In this paper, for the first time, we study the complexity/algorithmic
aspects of the famous protein local structure alignment problem under
the discrete \frechet\ distance. We show that the general problem is
NP-complete; in fact, it is even NP-hard to approximate within a factor of
$n^{1-\epsilon}$. On the other hand, when a constant number of proteins
are given then the problem can be solved in polynomial time. It would be
interesting to see the empirical comparisons of protein local structure
alignment under the discrete \frechet\ distance with the existing methods.
Another open problem, obviously, is whether it is possible to improve the
running time of the dynamic programming algorithms in Section 4.

\end{document}